\documentclass[a4paper,DIV12,11pt]{scrartcl}

\setkomafont{title}{\normalfont \LARGE}
\setkomafont{section}{\normalfont \Large \bfseries \boldmath}
\setkomafont{paragraph}{\normalfont \bfseries \boldmath}
\setkomafont{subsection}{\normalfont \large \bfseries\boldmath}
\setkomafont{subsubsection}{\normalfont \bfseries\boldmath}

\usepackage{amsmath,amsfonts,amssymb,amsthm}
\usepackage{cite}

\setkomafont{title}{\normalfont \LARGE}
\setkomafont{section}{\normalfont \Large \bfseries \boldmath}
\setkomafont{paragraph}{\normalfont \bfseries \boldmath}

\theoremstyle{remark}
\theoremstyle{plain}
\newtheorem{lemma}{Lemma}[section]
\newtheorem{theorem}[lemma]{Theorem}
\newtheorem{definition}[lemma]{Definition}
\newtheorem{remark}[lemma]{Remark}
\newtheorem{question}[lemma]{Question}
\newcommand{\kunsat}{\ensuremath{k\mathsf{UNSAT}}}
\newcommand{\threesat}{\ensuremath{3\mathsf{SAT}}}

\newcommand{\gcol}[1]{\ensuremath{k\mathtt{\text{-}Coloring}}}

\DeclareMathOperator{\DTime}{\mathsf{DTime}}
\DeclareMathOperator{\NTime}{\mathsf{NTime}}
\DeclareMathOperator{\poly}{poly}

\DeclareMathOperator{\probab}{\mathbb{P}}
\DeclareMathOperator{\expected}{\mathbb{E}}
\DeclareMathOperator{\supp}{supp}

\newcommand{\tiling}{\ensuremath{\mathtt{Tiling}}}

\newcommand{\nat}{\ensuremath{\mathbb{N}}}
\newcommand{\Dist}{\mathcal{D}}
\newcommand{\eps}{\varepsilon}
\newcommand{\SmoothedP}{\ensuremath{\mathsf{Smoothed\text{-}P}}}
\newcommand{\SmoothedPBM}{\ensuremath{\mathsf{Smoothed\text{-}PBM}}}
\newcommand{\SmoothedBPP}{\ensuremath{\mathsf{Smoothed\text{-}BPP}}}
\newcommand{\SmoothedRP}{\ensuremath{\mathsf{Smoothed\text{-}RP}}}
\newcommand{\AvgP}{\ensuremath{\mathsf{Avg\text{-}P}}}
\newcommand{\AvgPBM}{\ensuremath{\mathsf{Avg\text{-}PBM}}}

\newcommand{\ds}[1]{\ensuremath{#1_{\mathrm{ds}}}}
\newcommand{\oblds}[1]{\ensuremath{#1_{\mathrm{ds}}^{\mathrm{obl}}}}
\newcommand{\les}{\le_{\mathrm{smoothed}}}
\newcommand{\BH}{\ensuremath{\mathtt{BH}}}
\newcommand{\PCompSmo}{\mathsf{PComp}_{\mathrm{para}}}
\newcommand{\UBH}{U^{\mathtt{BH}}} 
\newcommand{\bin}{\operatorname{bin}}

\newcommand{\DistNP}{\ensuremath{\mathsf{DistNP}}}
\newcommand{\DP}{\ensuremath{\mathsf{P}}}
\newcommand{\NP}{\ensuremath{\mathsf{NP}}}
\newcommand{\RP}{\ensuremath{\mathsf{RP}}}
\newcommand{\EXPLIN}{\ensuremath{\mathsf{E}}}
\newcommand{\NEXPLIN}{\ensuremath{\mathsf{NE}}}
\newcommand{\DistNPSmo}{\ensuremath{\mathsf{Dist\text{-}NP}_{\mathrm{para}}}}

\title{Smoothed Complexity Theory\thanks{An extended abstract of this paper will appear in the \emph{Proceedings of
the 37th Int.\ Symp.\ on
Mathematical Foundations of Computer Science (MFCS 2012)}. Supported by DFG research grant BL 511/7-1.}}

\author{Markus Bl\"aser$^1$ \and Bodo Manthey$^2$}

\date{\small $^1$Saarland University, Department of Computer Science, \texttt{mblaeser@cs.uni-saarland.de} \\
$^2$University of Twente, Department of Applied Mathematics, \texttt{b.manthey@utwente.nl}}

\begin{document}

\maketitle
\thispagestyle{plain}
\pagestyle{plain}

\begin{abstract}
Smoothed analysis is a new way of analyzing algorithms
introduced by Spielman and Teng (\emph{J. ACM}, 2004). 
Classical methods like worst-case or average-case analysis
have accompanying complexity classes, like $\DP$ and $\AvgP$,
respectively. While worst-case or average-case analysis 
give us a means to talk about the running time of a particular
algorithm, complexity classes allows us to talk
about the inherent difficulty of problems.

Smoothed analysis is a hybrid of worst-case and average-case analysis
and compensates some of their drawbacks. Despite its success for the analysis of
single algorithms and problems, there is no embedding of smoothed analysis
into computational complexity theory, which is necessary to classify problems
according to their intrinsic difficulty.

We propose a framework for smoothed complexity theory, define
the relevant classes, and prove some first hardness results (of bounded halting and tiling)
and tractability results (binary optimization problems, graph coloring, satisfiability).
Furthermore, we discuss extensions and shortcomings of our model and relate it to semi-random models.
\end{abstract}

\section{Introduction}

The goal of computational complexity theory is to classify
computational problems according to their intrinsic 
difficulty. While the analysis of algorithms
is concerned with analyzing, say,  the running time of
a particular algorithm, complexity theory rather analyses
the amount of resources that all algorithms need at least
to solve a given problem.

Classical complexity classes, like $\DP$, reflect
worst-case analysis of algorithms. A problem is
in $\DP$ if there is an algorithm whose running time
on all inputs of length $n$ is bounded by a polynomial
in $n$. Worst-case analysis has been a success 
story: The bounds obtained are valid for every input
of a given size, and, thus, we do not have to think about
typical instances of our problem. If an algorithm
has a good worst-case upper bound, then this is a very strong statement:
The algorithm will perform
well in practice. (For practical purposes, ``good upper bound'' of course also includes
constants and lower order terms.) 

However, some algorithms 
work well in practice
despite having a provably high worst-case running
time. The reason for this is that the worst-case
running time can be dominated by a few pathological
instances that rarely or never occur in practice.
An alternative to worst-case analysis is average-case analysis.
Many of the algorithms with poor worst-case but good practical performance have a good average running time.
This means that the expected running time with instances drawn according to
some fixed probability distribution is low.

In complexity-theoretic terms,
$\DP$ is the class of all problems 
that can be solved with polynomial worst-case running time.
In the same way, the class $\AvgP$ is the class
of all problems that have polynomial average-case
running time. Average-case complexity theory
studies the structural properties of average-case running time.
Bogdanov and Trevisan give a comprehensive survey of average-case complexity~\cite{AverageSurvey}.

While worst-case complexity has the drawback of being often pessimistic,
the drawback of average-case analysis is that random instances have often
very special properties with high probability.
These properties of random instances distinguish them from typical instances.
Since a random and a typical instance
is not the same, a good average-case running time does not
necessarily explain a good performance in practice.
In order to get a more realistic performance measure,
(and, in particular, to explain the speed of the simplex method),
Spielman and Teng have proposed a new way
to analyze algorithms called \emph{smoothed analysis}~\cite{Simplex}.
In smoothed analysis, an adversary chooses an instance, and then this
instance is subjected to a slight random perturbation.
We can think of this perturbation as modeling measurement errors or random noise or
the randomness introduced by taking, say, a random poll.
The perturbation is controlled by some parameter $\phi$, called the \emph{perturbation parameter}.
Spielman and Teng have proved that the simplex method
has a running time that is polynomial in the size of the
instance and the perturbation parameter~\cite{Simplex}. (More precisely, 
for any given instance,
the expected running time on the perturbed instance is bounded by a polynomial.)
Since then, the framework of smoothed analysis has been applied
successfully to a variety of algorithms
that have a good behavior in
practice (and are therefore widely used)
but whose worst-case running time indicates poor
performance~\cite{kMeans,LTRM,RandomKnapsack,Hoare,MoitraODonnell,Perceptron,Termination,SmoothedHirsch,WADS,2Opt}.
We refer to two recent surveys for a broader picture of smoothed analysis~\cite{SmoothedSurvey,SmoothedSurvey2}.

However, with only few exceptions~\cite{SmoothedBinary,SmoothedInteger}, smoothed analysis
has only been applied yet to single algorithms or single problems.
Up to our knowledge, there is currently no attempt to formulate a smoothed complexity
theory and, thus, to embed smoothed analysis into computational complexity.

This paper is an attempt to define a smoothed complexity theory, including notions of intractability,
reducibility, and completeness.
We define the class \SmoothedP\ (Section~\ref{sec:model}), which corresponds to problems that can be solved
smoothed efficiently,
we provide a notion of reducibility (Section~\ref{sec:reduction}),
and define the class \DistNPSmo, which is a smoothed analogue of \NP, and prove that it contains
complete problems (Section~\ref{sec:paradist}).
We continue with some basic observations (Section~\ref{sec:simple}).
We also add examples of problems in \SmoothedP\ (Sections~\ref{sec:integer} and~\ref{sec:graph}) and discuss the relationship of smoothed complexity to semi-random
models (Section~\ref{sec:semirandom}).
Finally, since this is an attempt of a smoothed complexity theory, we conclude with
a discussion of extension, shortcomings, and difficulties of our definitions (Section~\ref{sec:disc}).

\section{Smoothed Polynomial Time and \SmoothedP}
\label{sec:model}

\subsection{Basic Definitions}

In the first application of smoothed analysis to the simplex method~\cite{Simplex},
the strength of the perturbation has been controlled in terms of the standard deviation
$\sigma$ of the Gaussian perturbation.
While this makes sense for numerical problems, this model cannot be used
for general (discrete problems).
A more general form of perturbation models has been introduced by Beier and V\"ocking~\cite{RandomKnapsack}:
Instead of specifying an instance that is afterwards perturbed (which can also be viewed as the adversary
specifying the mean of the probability distribution according to which the instances are drawn), the adversary
specifies the whole probability distribution.
Now the role of the standard deviation $\sigma$ is taken over by the parameter $\phi$, which is an upper
bound for the maximum density of the probability distributions. For Gaussian perturbation,
we have $\sigma = \Theta(1/\phi)$. Because we do not want to restrict our theory to numerical
problems, we have decided to use the latter model.

Let us now define our model formally.
Our perturbation models are families 
of distributions $\Dist = (D_{n, x, \phi})$.
The length of $x$ is $n$ (so we could omit the index $n$
but we keep it for clarity). Note that length does not necessarily mean bit length, but depends
on the problem. For instance, it can be the number of vertices of the graph encoded by $x$.
For every $n$, $x$, and $\phi$,
the support of the distribution $D_{n,x,\phi}$ should be contained in the set
$\{0,1\}^{\le \poly(n)}$.
Let
\[
S_{n,x} = \bigl\{ y \mid  \text{$D_{n,x,\phi}(y) > 0$ for some $\phi$} \bigr\},
\] 
and let $N_{n,x} = |S_{n,x}|$.

For all $n$, $x$, $\phi$, and $y$,
we demand $D_{n,x,\phi}(y) \le \phi$.
This controls the strength of the perturbation
and restricts the adversary.
We allow $\phi \in [1/N_{n,x}, 1]$.
Furthermore, the values of $\phi$
are discretized, so that they can be described by at most
$\poly(n)$ bits.
The case
$\phi = 1$ corresponds to the worst-case complexity;
we can put all the mass on one string.
The case $\phi = 1/N_{n,x}$ models the average case;
here we usually have to put probability on an exponentially
large set of strings.
In general, the larger $\phi$, the more powerful the adversary.
We call such families $(D_{n,x,\phi})_{n,x,\phi}$ of probability distributions
\emph{parameterized families of distributions}.

Now we can specify what it means that an algorithm has smoothed polynomial running-time.
The following definition can also be viewed as a discretized version of
Beier and V\"ocking's definition~\cite{SmoothedBinary}.
Note that we do not speak about expected running-time, but about expected running-time
to some power $\eps$. This is because the notion of expected running-time
is not robust with respect to, e.g., quadratic slowdown.
The corresponding definition for average-case complexity is due to Levin~\cite{Levin}.
We refer to Bogdanov and Trevisan~\cite{AverageSurvey}
for a thorough discussion of this issue.

\begin{definition} \label{def:smoothed:1}
An algorithm $A$ has \emph{smoothed
polynomial running time} with respect
to the family $\Dist$ if there exists an $\eps > 0$ such that, for all $n$, $x$, and $\phi$, 
we have
\[\expected_{y \sim D_{n, x, \phi}} \bigl(t_A(y; n, \phi)^\eps\bigr)
   = O\bigl( n \cdot N_{n,x} \cdot \phi \bigr).
\]
\end{definition}

This definition implies that (average-)polynomial time is only required
if we have $\phi = O({\poly(n)}/N_{n,x})$. 
This seems to be quite
generous at first glance, but it is in accordance with,
e.g., Spielman and Teng's analysis of the simplex method~\cite{Simplex} or Beier and V\"ocking's
analysis of integer programs~\cite{SmoothedBinary};
they achieve polynomial time running time
only if they perturb all but at most $O(\log n)$ digits:
If we perturb a number with, say, a Gaussian of standard deviation $\sigma = 1/\poly(n)$,
then we expect that the $O(\log n)$ most significant bits remain untouched,
but the less significant bits are random.

In average-case complexity, one considers not decision problems alone, but decision problems
together with a probability distribution.
The smoothed analogue of this is that we consider
tuples $(L,\Dist)$, where $L \subseteq  \{0,1\}^\ast$ is
a decision problem and $\Dist$ is a parameterized family of distributions. We call such problems
\emph{parameterized distributional problems}.
The notion of smoothed polynomial running-time (Definition~\ref{def:smoothed:1}) allows us to define
what it means for a parameterized distributional problem to have polynomial smoothed complexity.

\begin{definition}
\label{def:smoothedp}
$\SmoothedP$ is the class of all $(L,\Dist)$ such that 
there is a deterministic algorithm $A$ 
with smoothed polynomial running time that decides $L$.
\end{definition}

We start with an alternative characterization of
smoothed polynomial time as it is known for average case
running time.
It basically says that an algorithm has smoothed polynomial running-time
if and only if its running-time has polynomially decreasing tail bounds.
Though smoothed polynomial time
is a generalization of average case polynomial time,
the characterization and the proof of equivalence
are similar.

\begin{theorem}\label{thm:alt1}
An algorithm $A$ has smoothed polynomial running time if and only if
there is an $\eps > 0$ and a polynomial $p$ such that
for all $n$, $x$, $\phi$, and $t$,
\[
\Pr_{y \sim D_{n,x,\phi}}[t_A(y;n,\phi) \ge t] 
         \le \frac {p(n)}{t^\eps} \cdot N_{n,x} \cdot \phi.
\]
\end{theorem}

\begin{proof}
Let $A$ be an algorithm whose running time $t_A$ fulfills
\[ \expected_{y \sim D_{n, x, \phi}} \bigl(t_A(y; n, \phi)^\eps\bigr)
   = O\left( n \cdot {N_{n,x}} {\phi} \right) .
\]
The probability that the running time exceeds a certain value $t$ can be bounded by
Markov's inequality:
\begin{align*}
  \probab (t_A(y; n, \phi)\ge t) & = 
    \probab \bigl(t_A(y; n, \phi)^\eps \ge t^\eps \bigr) \\
&    \le \frac{\expected_{y \sim D_{n, x, \phi}}\bigl(t_A(y; n, \phi)^\eps \bigr)}
             {t^\eps}
    = O\left(n \cdot N_{n,x} \phi \cdot t^{-\eps} \right) .
\end{align*}
For the other direction, assume that 
\[ \Pr_{y \sim D_{n,x,\phi}}[t_A(y;n,\phi) \ge t] 
         \le \frac {n^c}{t^\eps} \cdot N_{n,x} \phi
\]
for some constants $c$ and $\eps$.
Let $\eps' = \eps / (c + 2)$. Then we have
\begin{align*}
  \expected_{y \sim D_{n, x, \phi}} \bigl(t_A(y; n, \phi)^{\eps'}\bigr)
     & = \sum_{t} \probab (t_A(y; n , \phi)^{\eps'} \ge t) \\
     & \leq n + \sum_{t \ge n}  \probab (t_A(y; n , \phi) \ge t^{1/\eps'}) \\
     & \leq n + \sum_{t \ge n} {t^{-2}} \cdot N_{n,x} \phi
      = n + O(N_{n,x} \phi).
\end{align*}
\end{proof}

\subsection{Heuristic Schemes}
\label{ssec:heuristic}

A different way to think about efficiency in the smoothed setting is via so-called
heuristic schemes. This notion comes from average-case complexity~\cite{AverageSurvey}, but
can be adapted to our smoothed setting. The notion of a heuristic scheme comes from the observation that,
in practice, we might only be able to run our algorithm for a polynomial number of steps.
If the algorithms does not succeed within this time bound, then it ``fails'', i.e.,
it does not solve the given instance. The failure probability decreases polynomially
with the running time that we allow. The following definition captures this.

\begin{definition}
Let $(L,\Dist)$ be a smoothed distributional problem.
An algorithm $A$ is an errorless heuristic scheme for $(L,\Dist)$
if there is a polynomial $q$ such that
\begin{enumerate}
\item For every $n$, every $x$, every $\phi$, 
every $\delta > 0$, and every $y \in \supp D_{n,x,\phi}$, we have
$A(y;n,\phi,\delta)$ outputs either $L(y)$ or $\bot$.
\item For every $n$, every $x$, every $\phi$, 
every $\delta > 0$, and every $y \in \supp D_{n,x,\phi}$, we have
$t_A(y;n,\delta) \le q(n,N_{n,x} \phi,1/\delta)$.
\item For every $n$, $x$, $\phi$, $\delta > 0$, and $y \in \supp D_{n,x,\phi}$,
we have
$\Pr_{y \sim D_{n,x,\phi}}[A(y;n,\phi,\delta) = \bot] \le \delta$.
\end{enumerate}
\end{definition}

With the definition of a heuristic scheme, we can prove that heuristic schemes precisely
characterize \SmoothedP.

\begin{theorem}
\label{thm:heuristic}
$(L,\Dist) \in \SmoothedP$ if and only if $(L,\Dist)$ has
an errorless heuristic scheme.
\end{theorem}

\begin{proof}
Let $A$ be an algorithm for $(L,\Dist)$.
By Theorem~\ref{thm:alt1}, 
the probability that
\[
  \probab (t_A(y; n, \phi)\ge t) 
    = O\left(n \cdot N_{n,x} \phi \cdot t^{-\eps} \right) .
\]
We get an errorless heuristic scheme $B$ 
from $A$ as follows: Simulate $A$ for 
$(n \cdot N_{n,x} \phi / \delta)^{1/\eps}$ steps. 
If $A$ stops within these number of steps, then output whatever
$A$ outputs. Otherwise, output $\bot$. By the choice
of the parameters, the probability that $B$ outputs $\bot$
is bounded by $\delta$. 

For the other direction, let $A$ be an errorless heuristic scheme
for $(L,\Dist)$. We get an algorithm with smoothed
polynomial running time by first running $A$ with $\delta = 1/2$,
then with $\delta = 1/4$, and in the $i$th iteration 
with $\delta = 1/2^i$. Whenever $A$ does not answer $\bot$,
$B$ gives the same answer and stops. $B$ will eventually
stop, when $\delta < D_{n,x,\phi}(y)$.
For $i$ iterations, $B$ needs
\[
   \sum_{j = 1}^i q(n,N_{n,x}\phi, 2^j) 
      \le \poly(n, N_{n,x}\phi) \cdot 2^{ci}
\]
for some constant $c$.  $B$ stops after $i$ iterations
for all but a $2^{-i}$ fraction of the input. Thus
$B$ has smoothed polynomial running time. (Choose $\eps < 1/c$.)
\end{proof}

\subsection{Alternative Definition: Bounded Moments}
\label{ssec:boundedmoments}

At first glance, one might be tempted to use ``expected running time'' for the definition
of \AvgP\ and \SmoothedP. However, as mentioned above, simply using the expected running time
does not yield a robust measure. This is the reason why the expected value of the running time raised
to some (small) constant power is used.
R\"oglin and Teng~\cite[Theorem 6.2]{RoeglinTeng} have shown that for integer programming
(more precisely, for binary integer programs with a linear objective function),
the expected value indeed provides a robust measure.
They have proved that
a binary optimization problem can be solved in expected polynomial time if and only
if it can be solved in worst-case pseudo-polynomial time.
The reason for this is that
all finite moments of the Pareto curve are polynomially bounded. Thus, a polynomial slowdown
does not cause the expected running time to jump from polynomial to exponential.

As far as we are aware, this phenomenon, i.e., the case that all finite moments
have to be bounded by a polynomial, has not been studied yet in average-case complexity.
Thus, for completeness, we define the corresponding average-case and smoothed complexity classes
as an alternative to \AvgP\ and \SmoothedP.

\begin{definition}
\begin{enumerate}
\item An algorithm has \emph{robust smoothed
polynomial running time} with respect to $\Dist$ if, for all fixed $\eps > 0$ and for every $n$, $x$, and $\phi$, 
we have
\[
\expected_{y \sim D_{n, x, \phi}} \bigl(t_A(y; n, \phi)^\eps\bigr)
   = O\bigl( n \cdot N_{n,x} \cdot \phi \bigr).
\]
\SmoothedPBM\ is the class of all $(L, \Dist)$ for which there exists a deterministic
algorithm with robust smoothed polynomial running time. (The ``PBM'' stands for ``polynomially bounded moments''.)

\item An algorithm $A$ has \emph{robust average
polynomial running time} with respect to $\Dist$ if, for all fixed $\eps > 0$ and for all $n$,
we have
$\expected_{y \sim D_{n}} \bigl(t_A(y)^\eps\bigr)
   = O( n )$.
\AvgPBM\ contains all $(L, \Dist)$ for which there exists a deterministic
algorithm with robust smoothed polynomial running time.
\end{enumerate}
\end{definition}

From the definition, we immediately get $\SmoothedPBM \subseteq \SmoothedP$ and $\AvgPBM \subseteq \AvgP$.
Moreover, if $L \in \DP$, then $L$ together with any family of distributions is also in \SmoothedP\ and \AvgP\
and also in \SmoothedPBM\ and \AvgPBM.
From R\"oglin and Teng's result~\cite{RoeglinTeng}, one might suspect
$\AvgP = \AvgPBM$ and $\SmoothedP = \SmoothedPBM$,
but this does not hold.

\begin{theorem}
\label{thm:separation}
$\AvgPBM \subsetneq \AvgP$ and $\SmoothedPBM \subsetneq \SmoothedP$.
\end{theorem}

\begin{proof}
We only prove the theorem for average-case complexity. The proof
for the smoothed complexity case is almost identical.

By the time hierarchy theorem~\cite{AroraBarak}, there is a language
$L' \in \DTime(2^{n})$ such that $L' \notin \DTime(2^{o(n)})$.
Consider the following language $L = \{x0^{n} \mid |x| = n, x \in L'\}$.
Let $\Dist' = (D'_n)$ be a hard probability distribution for $L'$, i.e.,
$(L', \Dist')$ is as hard to solve as $L'$ in the worst case~\cite{LiVit,LiVitBook}

Let $\Dist = (D_n)$ be given as follows:
\[
  D_n(xy) = \begin{cases}
   2^{-n} \cdot D'_n(x) & \text{if $|x| = n$ and $y = 0^n$ and} \\
   2^{-2n} & \text{otherwise.}
  \end{cases}
\]
Since $L' \in \DTime(2^n)$, we have $(L, \Dist) \in \AvgP$:
$L$ can be decided in expected time time $2^{-n} \cdot 2^n + O(n) = O(n)$.
Now we prove that $(L, \Dist) \notin \AvgPBM$.
If $(L, \Dist)\in \AvgPBM$ were true, then 
$2^{-n} \cdot \expected_{x \sim D'_n} (t(x)^c)$ would be bounded by a polynomial
for all fixed $c$. Here, $t$ is the time needed to solve the $L'$ instance $x$.

Our choice of $\Dist'$, Jensen's inequality, and the fact that $L' \notin \DTime(2^{o(n)})$
imply that $\expected_{x \sim D'_n} (t(x)^c) \geq
\expected_{x \sim D'_n} (t(x))^c = 2^{c \cdot \Omega(n)}$.
Thus, for some sufficiently large $c$,
$2^{-n} \cdot \expected_{x \sim D'_n} (t(x)^c)$ exceeds any polynomial.
\end{proof}

\begin{question}
A natural question to ask is:
Does there exist a language $L \in \NP$ together with some ensemble $\Dist$
such that $(L, \Dist)$
separates $\AvgP$ from $\AvgPBM$ and $\SmoothedP$ from $\SmoothedPBM$?
Does there exist some $L$ together with a computable ensemble $\Dist$
that separates these classes?
\end{question}
We conjecture that, assuming the exponential time hypothesis (ETH)~\cite{ETH}, such an $L \in \NP$
exists to separate $\AvgP$ from $\AvgPBM$ and $\SmoothedP$ from $\SmoothedPBM$. Given the ETH, $\threesat$
requires time $2^{\Omega(n)}$ in the worst case,
thus also on average if we use the universal distribution. This holds even if we restrict
$\threesat$ to $O(n)$ clauses. However, $n$ is here the number of variables, not the bit length
of the input, which is roughly $\Theta(n \log n)$. Thus, a direct application of the ETH seems to
be impossible here.

\section{Disjoint Supports and Reducibility}
\label{sec:reduction}

The same given input $y$ can appear with very high and with
very low probability at the same time. What sounds like a
contradiction has an easy explanation: $D_{n,x,\phi}(y)$ can
be large whereas $D_{n,x',\phi}(y)$ for some $x' \neq x$ is small.
But if we only see $y$, we do not know whether $x$ or $x'$ was
perturbed. This causes some problems when one wants to develop a
notion of reduction and completeness.

For a parameterized distributional problem $(L,\Dist)$, let
\[
   \ds L = \{ \langle x,y \rangle \mid 
              \text{$y \in L$ and $|y| \le \poly(|x|)$} \}.
\]
The length of $|y|$ is bounded by the same polynomial
that bounds the length of the strings in any $\supp D_{n,x,\phi}$.
We will interpret a pair $\langle x,y\rangle$ as
``$y$ was drawn according to $D_{n,x,\phi}$''.
With the notion of \ds L, we can now define a reducibility between
parameterized distributional problems.
We stress that, although the definition below involves \ds L\ and \ds{L'},
the reduction is defined for pairs $L$ and $L'$ and neither of the two is required to
be a disjoint-support language. This means that, for $(L, \Dist)$, the supports of $D_{n,x,\phi}$
for different $x$ may intersect. And the same is allowed for $(L', \Dist')$.

\begin{definition}
\label{def}
Let $(L,\Dist)$ and $(L',\Dist)$ be two parameterized distributional
problems. $(L,\Dist)$ reduces to $(L',\Dist')$
(denoted by ``$(L,\Dist) \les (L',\Dist')$'') if there
is a polynomial time computable function $f$ 
such that for every $n$, every $x$, every $\phi$
and every $y \in \supp D_{n,x,\phi}$
the following holds:
\begin{enumerate}
\item $\langle x,y \rangle \in \ds L$ if and only if 
$f(\langle x,y \rangle;n,\phi) \in \ds L'$.
\item There exist polynomials $p$ and $m$ such that, for
every $n$, $x$, and $\phi$ and
every $y' \in \supp D'_{m(n),f_1(\langle x,y \rangle ;n,\phi),\phi}$,
we have
\[
\textstyle
\sum_{y: f_2(\langle x, y \rangle ;n,\phi) = y'}
      D_{n,x,\phi}(y) \le p(n) D_{m(n),f_1(\langle x, y \rangle ;n,\phi),\phi}(y'),
\]
where $f(\langle x,y \rangle;n,\phi)
= \langle f_1 (\langle x,y \rangle;n,\phi),
          f_2 (\langle x,y \rangle;n,\phi) \rangle$.
          \label{item2}
\end{enumerate}
\end{definition}

\begin{remark}
Note that we could also allow that $\phi$ on the right-hand side is polynomially
transformed. However, we currently do not see how to benefit from this.
\end{remark}

It is easy to see that $\les$ is transitive.
Ideally, $\SmoothedP$ should be closed under this type of reductions.
However, we can only show this for the related class
of problems with disjoint support.

\begin{definition}
$\ds \SmoothedP$ is the set of all distributional
problems with disjoint supports such that there is an
algorithm $A$ for $\ds L$ with smoothed polynomial running time.
(Here, the running time on $\langle x,y \rangle$ is
defined in the same way as in Definition~\ref{def:smoothed:1}.
Since $|y| \le \poly(|x|)$ for a pair $\langle x,y\rangle \in \ds L$,
we can as well measure the running time in $|x|$.)
\end{definition}

Now, $\ds \SmoothedP$\ is indeed closed under the above type of reductions.

\begin{theorem}
\label{thm:closed}
If $(L,\Dist) \les (L',\Dist')$ and $(\ds L',\Dist') \in \ds \SmoothedP$,
then $(\ds L,\Dist) \in \ds \SmoothedP$.
\end{theorem}

\begin{proof}
Let $A'$ be a an errorless heuristic scheme for $(\ds L',\Dist')$.
Let $f$ be the reduction from $(L,\Dist)$ to $(L',\Dist')$
and let $p$ and $m$ be the corresponding polynomials.

We claim that $A(\langle x, y \rangle ;n,\phi,\delta) 
= A'(f(\langle x,y \rangle;n,\phi);m(n),\phi,\delta/p(n))$
is an errorless heuristic scheme for $(\ds L,\Dist)$.
To prove this, let 
\[
B = \{y' \in \supp D'_{m(n),f_1(\langle x,y \rangle ;n,\phi),\phi} \mid
A'(\langle f(\langle x,y \rangle;n,\phi),y'\rangle ;m(n),\phi,\delta/p(n)) = \bot\}
\]
be the set of string
on which $A'$ fails.

Because $A'$ is a heuristic scheme, we have
$D'_{m(n),f_1(\langle x,y\rangle ;n,\phi),\phi}(B) \le \delta / p(n)$.
Therefore,
\begin{align*}
 &\qquad \Pr_{y \sim D_{n,x,\phi}} (A(\langle x,y \rangle;n,\phi,\delta) = \bot) \\
  &   = \Pr_{y \sim D_{n,x,\phi}}(A'(f(\langle x,y \rangle ;n,\phi);m(n),\phi,\delta/p(n) = \bot) \\
    & = \sum_{y: f_2(\langle x,y \rangle ;n,\phi) \in B} D_{n,x,\phi}(y) \\
    & \le \sum_{y' \in B} p(n) D'_{m(n);f_1(\langle x,y \rangle ;n,\phi);\phi} (y') \\
    & = p(n) D'_{m(n);f_1(\langle x,y \rangle ;n,\phi);\phi} (B)  \le \delta.
\end{align*}
Thus, $(\ds L,\Dist) \in \ds \SmoothedP$.
\end{proof}

With the definition of disjoint support problems, a begging question is
how the complexity of $L$ and $\ds L$ are related.
It is obvious that $(L, \Dist) \in \SmoothedP$ implies
$(\ds L, \Dist) \in \ds \SmoothedP$. However, the converse is
not so obvious.
The difference between $L$ and $\ds L$ is
that for $\ds L$, we get the $x$ from which the
input $y$ was drawn. While this extra information
does not seem to be helpful at a first glance, 
we can potentially use it to extract randomness from
it. So this question is closely related
to the problem of derandomization.

But there is an important subclass of problems in
$\ds \SmoothedP$ whose counterparts are in $\SmoothedP$, namely
those which have an oblivious algorithm with smoothed
polynomial running time. We call an algorithm
(or heuristic scheme) for some problem with disjoint 
supports \emph{oblivious} if the running
time on $\langle x,y \rangle$ does not depend on $x$ 
(up to constant factors). Let $\oblds \SmoothedP$ be
the resulting subset of problems in $\ds \SmoothedP$ that have such
an oblivious algorithm with smoothed polynomial running time.

\begin{theorem} 
\label{thm:oblivious}
For any parameterized problem $(L,\Dist)$,
$(L,\Dist) \in \SmoothedP$ if and only if $(\ds L,\Dist) \in \oblds \SmoothedP$.
\end{theorem}

\begin{proof}
Let $A$ be an oblivious algorithm with smoothed polynomial running
time for $\ds L$. Since $A$ is oblivious,
we get an algorithm for $L$
with the same running time (up to constant factors) by
running $A$ on $\langle x_0,y \rangle$ on input $y$,
where $x_0$ is an arbitrary string of length $n$.
\end{proof}

Note that almost all algorithms, for which a 
smoothed analysis has been carried out,
do not know the $x$ from which $y$ was drawn; in particular,
there is an oblivious algorithm for them.
\begin{question}
Thus, we ask the question:
Is there a problem $(L,\Dist) \notin \SmoothedP$
but $(\ds L,\Dist) \in \ds \SmoothedP$?
\end{question}

Note that in $\ds L$, each $y$ is paired with \emph{every} $x$,
so there is no possibility to encode information
by omitting some pairs. This prohibits attempts  for constructing
such a problem
like considering pairs $\langle x,f(x) \rangle$
where $f$ is some one-way function.
However, a pair $\langle x,y\rangle$ contains randomness
that one could extract.
For the classes $\SmoothedBPP$
or $\SmoothedP/\mathsf{poly}$, which can be defined
in the obvious way, knowing $x$ does not seem to help.
It should be possible to use the internal random bits (or the advice)
to find an $x'$ that is good enough.

\section{Parameterized Distributional $\NP$}
\label{sec:paradist}

\subsection{\DistNPSmo}

In this section, we define the smoothed analogue of the worst-case class \NP\
and the average-case class \DistNP~\cite{Levin,Gurevich}.
First, we have to restrict ourself to ``natural'' distributions.
This rules out, for instance, probability distributions
based on Kolmogorov complexity that (the \emph{universal distribution}),
under which worst-case complexity equals average-case complexity for all problems~\cite{LiVit}.
We transfer the notion of computable ensembles to smoothed complexity.

\begin{definition}
A parameterized family of distributions is in 
$\PCompSmo$ if the cumulative probability
\[
   F_{D_{n,x,\phi}} = \sum_{z \le x} D_{n,x,\phi}
\]
can be computed in polynomial time (given $n$, $x$ and $\phi$ 
in binary).
\end{definition}

With this notion, we can define the smoothed analogue of \NP\ and \DistNP.

\begin{definition}
$\DistNPSmo = \{(L,\Dist) \mid \text{$L\in \NP$ and $\Dist \in \PCompSmo$} \}$.
\end{definition}

\subsection{\DistNPSmo-Complete Problems}
\label{ssec:complete}

\subsubsection{Bounded Halting}

Having defined \DistNPSmo\
in the previous section,
we now prove that
\emph{bounded halting} -- given a Turing machine, an input, and a running-time
bound, does the Turing machine halt on this input within the given time bound -- is complete
for \DistNPSmo. Bounded halting is the canonical \NP-complete language, and it has been
the first problem that has been shown to be \AvgP-complete~\cite{Levin}.
Formally, let
\[
  \BH = \{\langle g,x,1^t\rangle  \mid \text{NTM with G\"odel number $g$  
          accepts $x$ within $t$ steps}\}.
\]

In order to show that \BH\ is \DistNPSmo-complete, we need a ``compression function''
for probability distributions~\cite{AverageSurvey}. This is the purpose of the following lemma.

\begin{lemma}
\label{lem:compression}
Let $\Dist = (D_{n,x,\phi}) \in \PCompSmo$ be an ensemble. There
exists a deterministic algorithm $C$ such that the following holds:
\begin{enumerate}
\item $C(y;n,x,\phi)$ runs in time polynomial in $n$ and $\phi$
for all $y \in \supp D_{n,x,\phi}$.
\item For every $y,y' \in \supp D_{n,x,\phi}$,
$C(y;n,x,\phi) = C(y';n,x,\phi)$ implies $y = y'$.
\item If $D_{n,x,\phi}(y) < 2^{-|y|}$, then
$|C(y;n,x,\phi)| = 1 + |y|$. Else,
$|C(y;n,x,\phi)| 
     = \log \frac 1{D_{n,x,\phi}(y)} + c \cdot \log n + 1$
for some constant $c$.
\end{enumerate}
\end{lemma}

\begin{proof}
Consider any string $y \in \supp(D_{n,x,\phi})$.
If $D_{n,x,\phi}(y) \leq 2^{-|y|}$, then
we let $C(y;n,x,\phi) = 0y$.
If $D_{n,x,\phi}(y) > 2^{-|x|}$, then let $y'$ be the string
that precedes $y$ in lexicographic order, and let $p=F_{D_{n,x,\phi}}(y')$.
Then we set $C(y;n,x,\phi) = 1a$, where $a$ is the longest common prefix
of the binary representation of $p$ and $F_{D_{n,x,\phi}}(y) = p+D_{n,x,\phi}(y)$.
Since $\Dist \in \PCompSmo$, the string $z$ can be computed in polynomial time.
Thus, $C$ can be computed in polynomial time. (This also shows that
$|C(x,n)|$ is bounded by a polynomial in $|x|$.)

We have $D_{n,x,\phi}(y) \le 2^{-|a|}$, since adding $D_{n,x,\phi}(y)$ leaves
the first $|a|$ bits of $p$ unchanged.

Let $z$ be another string, $z'$ its predecessor and
$b$ the longest common prefix of $q = F_{D_{n,x,\phi}}(z')$
and $q + D_{n,x,\phi}(z')$. The intervals
$[p,p+D_{n,x,\phi}(y))$ and $[q,q+D_{n,x,\phi}(y))$ are
disjoint by construction. Therefore, $a$ and $b$
have to be different, because otherwise these intervals
would intersect. 

Let $c$ be such that $|y| \le n^c$
for all $y \in \supp(D_{n,x,\phi})$.
We set 
\[
  C(y;n,x,\phi) = 1 \bin(|a|) a 0^{\log \frac 1 {D_{n,x,\phi}(y)} - |a|}.
\] 
(Note that $\log \frac 1 {D_{n,x,\phi}(y)} \ge |a|$.)
Here $\bin(|a|)$ is a fixed length binary encoding of $a$.
We can bound this length by $c \log n$. The total length
of $C(y;n,x,\phi)$ is
\[
  |C(y;n,x,\phi)| = 1 + c \log n + \log \frac 1 {D_{n,x,\phi}(y)}.
\] 
It remains to be proved that $C$ is injective.
Let $C(y;n,x,\phi) = C(z;n,x,\phi)$. If $C(y;n,x,\phi)$ 
starts with a $0$, then obviously $y = z$. 
If $C(y;n,x,\phi)$ starts with a $1$, then the
prefixes $a$ and $b$ are the same. Therefore $y = z$
by the consideration above.
\end{proof}

The instances of $\BH$ are triples 
$\langle g,x,1^t\rangle$ of length 
$2 \log|g| + 2 \log |x| + |x| + |g| + t + \Theta(1)$.
Note that the instances of $\BH$ can be made prefix-free.
Let
\[
  \UBH_{N,\langle g,x,1^t\rangle,\phi}(\langle g',x',1^{t'} \rangle) = 
    \begin{cases}
    c_\phi \cdot
    2^{-|x'|} 
      & \text{if $g = g'$, $N = |\langle g',x',1^{t'} \rangle|$,
      and
      $|x'| \ge \log \frac 1{\phi}$,} \\
    0 & \text{otherwise.}
    \end{cases}
\]
Above, $c_\phi$ is an appropriate scaling factor.
More precisely, $c_\phi$ is the reciprocal of the number of possible lengths
for a string $x'$, i.e., it is of order $\frac 1{N - \log \phi}$.
In particular, $\UBH_{N,\langle g,x,1^t \rangle,\phi}(y) \le \phi$
for all $y$.

\begin{theorem}
\label{thm:complete}
$(\BH,\UBH)$ is $\DistNPSmo$-complete under polynomial-time
smoothed reductions
for some $\UBH \in \PCompSmo$.
\end{theorem}

\begin{proof}
Let $(L,\Dist) \in \DistNPSmo$ be arbitrary. 
Let $p(n)$ be an upper bound for the length
of the strings in any $\supp (D_{n,x,\phi})$.
Let $M$ be a nondeterministic machine that accepts
an input $a$ if and only if there is a string $y \in L$
with $C(y;n,x,\phi) = a$. Let $q$ be an upper bound on the
running time of $M$. Let $g$ be the G\"odel number of $M$.
Our reduction maps a string $y$ to
\[
   f(\langle x,y \rangle ;n,x,\phi) 
     = \left \langle \langle g,C(x;n,x,\phi),1^t \rangle,
                     \langle g,C(y;n,x,\phi),1^{t'} \rangle \right \rangle
\]
where $t$ and $t'$ chosen in such a way
that they are  larger than $q(p(|x|))$.
(And $t$ and $t'$ should be chosen in such a way
that all tuples have the same length $N(n)$.)

By construction, $\langle x,y \rangle \in \ds L$ if and only if 
$f(\langle x, y \rangle ;n,x,\phi) \in \ds \BH$. 

Domination remains to be verified. Since 
$C$ is injective, at most one $y$ is mapped 
to $\langle g,a,1^t \rangle$ given $n$, $x$, and $\phi$.
We have
\[
\UBH_{N,\langle g,C(x;n,x,\phi),1^{t} \rangle,\phi}(\langle g,C(y;n,x,\phi),1^{t'} \rangle)
     = c_\phi \cdot 2^{-|C(y;n,x,\phi)|}.
\]
If $|C(y;n,x,\phi)| \le \log \frac 1{D_{n,x,\phi}(y)} + c \log n + 1$,
then
\[
  \UBH_{N,\langle g,C(x;n,x,\phi),1^{t} \rangle, \phi}(\langle g,C(y;n,x,\phi),1^{t'} \rangle)
     \le c_\phi \cdot \frac{D_{n,x,\phi}(x)}{2 n^c}
\]
 and domination is fulfilled. If $|C(y;n,x,\phi)| = 1 + |y|$,
then 
\[
\UBH_{N,\langle g,C(x;n,x,\phi),1^{t} \rangle, \phi}(\langle g,C(y;n,x,\phi),1^{t'} \rangle)
    \le c_\phi \cdot 2^{-|y| - 1} \le 2 c_\phi \cdot D_{n,x,\phi}(y).
\]
This completes the hardness proof. The completeness follows since
$(\BH,\UBH)$ is indeed contained in $\DistNPSmo$.
\end{proof}

\subsubsection{Tiling}

The original \DistNP-complete problem by Levin~\cite{Levin} was \tiling\
(see also Wang~\cite{Wang}): An instance of the problem
consists of a finite set $T$ of square tiles, a positive integer $t$, and a sequence $s = (s_1, \ldots, s_n)$
for some $n \leq t$ such that $s_i$ matches $s_{i+1}$ (the right side of $s_i$ equals the left side of $s_{i+1}$).
The question is whether $S$ can be extended to tile an $n \times n$ square using tiles from $T$.

We use the following probability distribution for \tiling:
\[
  U^{\tiling}_{N,\langle T,s,1^t\rangle,\phi}(\langle T',s',1^{t'} \rangle) = 
    \begin{cases}
    c_\phi
 \cdot  
    a^{-|s'|} 
      & \text{if $T = T'$, $N = |\langle T',s',1^{t'} \rangle|$,
              $|T'| \ge \log \frac 1{\phi}$,} \\
    0 & \text{otherwise.}
    \end{cases}
\]
Here, $a$ is the number of possible choices in $T$ for each initial
tile $s_i$.

\begin{theorem}
\label{thm:tiling}
$(\tiling, U^{\tiling})$ is $\DistNPSmo$-complete
for some $U^{\tiling} \in \PCompSmo$
under polynomial-time smoothed reductions.
\end{theorem}

\begin{proof}
By construction, we have $(\tiling, U^{\tiling}) \in \DistNPSmo$.
For simplicity, we assume that the set $T$ of tiles always contains two tiles encoding
the input bits ``0'' and ``1'' and that these are the only possible tiles for
the initial tiling $(s_1, \ldots, s_n)$. (The problem does not become
easier without this restriction, but the hardness proof becomes more technical.)

For the hardness, $(\BH, \UBH)$ reduces to $(\tiling, U^{\tiling})$
because the Turing machine computations can be encoded as tiling problems in a straightforward
way~\cite{Wang} (the G\"odel number $g$ maps to some set $T$ of tiles, and the input $x$
maps to the initial tiling $s$). Finally, Item~\ref{item2} of the reduction (Definition~\ref{def}) is
fulfilled because of the similarity between the two probability distributions.
\end{proof}

\section{Basic Relations to Worst-Case Complexity}
\label{sec:simple}

In this section, we collect some simple facts about \SmoothedP\ and \DistNPSmo\ and their
relationship to their worst-case and average-case counterparts.
First, \SmoothedP\ is sandwiched between \DP\ and \AvgP, which
follows immediately from the definitions.

\begin{theorem}
If $L \in \DP$, then $(L, \Dist) \in \SmoothedP$ for any $\Dist$.
If $(L, \Dist) \in \SmoothedP$ with $\Dist = (D_{n,x,\phi})_{n, x, \phi}$, then
$(L, (D_{n,x_n,\phi})_{n}) \in \AvgP$
for $\phi = O(\poly(n)/N_{n,x})$ and every sequence $(x_n)_n$ of strings with $|x_n| \leq \poly(n)$.
\end{theorem}

Second, for unary languages, $\AvgP$ and $\DP$ coincide. The reason is that
for unary languages, we have just one single instance $1^n$ for each length $n$, and
this instance has a probability of $1$.
Also $\SmoothedP$ coincides with $\AvgP$ and $\DP$ for unary languages.
Because the set of instances is just a singleton, the parameter $\phi$ is fixed
to $1$ in this case.
The observation that $\AvgP$, $\DP$, and $\SmoothedP$ coincide for unary languages
allows us to transfer the result that $\DistNP \subseteq \AvgP$
implies $\NEXPLIN = \EXPLIN$~\cite{BenDavidEA} to smoothed complexity.
(The latter classes are defined as $\NEXPLIN = \NTime(2^{O(n)})$ and $\EXPLIN = \DTime(2^{O(n)})$.)
The transfer of their result to smoothed complexity is straightforward and
therefore omitted.

\begin{theorem}
\label{thm:bendavid}
If $\DistNPSmo \subseteq \SmoothedP$, then $\NEXPLIN = \EXPLIN$.
\end{theorem}

This gives some evidence that $\DistNPSmo$ is not a subset of $\SmoothedP$.
(We cannot have equality because 
\DistNPSmo\ is restricted to computable ensembles
and problems in \NP, while \SmoothedP\ does not have these restrictions.)

\section{Tractability 1: Integer Programming}
\label{sec:integer}

\label{sec:trac1}

Now we deal with tractable -- 
in the sense of smoothed complexity
-- optimization problems:
We show that if a binary integer linear program can be solved in pseudo-polynomial time,
then the corresponding decision
problem belongs to \SmoothedP. This result is similar to Beier and V\"ockings characterization~\cite{SmoothedBinary}:
Binary optimization problems
have smoothed polynomial complexity (with respect to continuous distributions)
if and only if they can be solved in
randomized pseudo-polynomial time.
We follow their notation and refer to their lemmas wherever appropriate.

\subsection{Setup and Probabilistic Model.}

A binary optimization problem is an optimization problem of the form
\[
\begin{array}{rl}
  \text{maximize} & c^T x \\
  \text{subject to} & w_i^T x \leq t_i \text{ for $i \in [k]$ and } \\
  &    x \in S \subseteq \{0,1\}^n.
     \end{array}
\]
Here, $c^T x = \sum_{j=1}^n c_j x_j$ is the linear objective function and
$w_i^T x = \sum_{j=1}^n w_{i,j} x_j \leq t_i$ are linear constraints.
Furthermore, we have the constraint that the binary vector $x$ must be contained in the set $S$.
This set $S$ should be viewed as containing the ``structure'' of the problem. Examples
are that $S$ contains all binary vectors representing spanning trees of a graph of $n$ vertices
or that $S$ represents all paths connecting two given vertices
or that $S$ contains all vectors corresponding to perfect matchings
of a given graph.
Maybe the simplest case is $k =1$ and $S = \{0,1\}^n$; then the binary program above
represents the knapsack problem.

We assume that $S$ is adversarial (i.e., non-random).
Since we deal with decision problems in this paper rather than with optimization
problems, we use the standard approach and introduce a threshold for the objective
function. This means that the optimization problem becomes the question whether
there is an $x \in S$ that fulfills $c^T x \geq b$ as well as
$w_i^T x \leq t_i$ for all $i \in \{1, \ldots, k\}$.
In the following, we treat the budget constraint $c^T x \geq b$
as an additional linear constraint for simplicity.
We call this type of problems \emph{binary decision problems}.

For ease of presentation, we assume that we have
just one linear constraint (whose coefficients will be perturbed) and
everything else is encoded in the set $S$. This means that the binary decision problem
that we want to solve is the following:
Does there exist an $x \in S$ with $w^T x \leq t$?

The values $w_1, \ldots, w_n$ are $n$-bit binary numbers.
Thus, $w_i \in \{0,1, \ldots, 2^{n}-1\}$.
While we can
of course vary their length, we choose to do it this way as it
conveys all ideas while avoiding another parameter.

Let us now describe the perturbation model.
We do not make any assumption about the probability distribution of any single coefficient.
Instead, our result holds for any family of probability distribution that fulfills
the following properties:
\begin{itemize}
\item
$w_1, \ldots, w_n$ are drawn according to independent distributions.
  The set $S$ and the threshold $t$ are part of the input and not subject to randomness.
  Thus, $N_{n, (S,w, t)} = 2^{n^2}$
  for any instance $(S, w, t)$ of size $n$.
  We assume that $S$ can be encoded by a polynomially long string.
  (This is fulfilled for most natural optimization problems, like TSP, matching, shortest
  path, or knapsack.)
\item
The fact that $w_1, \ldots, w_n$ are drawn independently
  means  that the probability for one coefficient to assume a specific value
  is bounded from above by $\phi^{1/n}$.
\end{itemize}
Since $N_{n, (S, w)} = 2^{n^2}$, the perturbation parameter
$\phi$ can vary between $2^{-n^2}$
(for the average case) and $1$ (for the worst case).

The idea is as follows: If we have a pseudo-polynomial algorithm, then
we can solve instances with $O(\log n)$ bits per coefficient efficiently.
Our goal is thus to show that $O(\log n)$ bits suffice with high probability.
(This is for the average case, i.e., $\phi = 2^{-n^2}$. For larger $\phi$,
more but not too many bits are needed.)
The proofs in the following are similar to proofs by Beier and V\"ocking~\cite{SmoothedBinary}.
However, at various places it gets slightly more technical because we have discrete
rather than continuous probability distributions.

The following simple lemma bounds the probability that a certain coefficient assumes a
value in a given small interval.

\begin{lemma}
\label{lem:boundeddensity}
Let $\delta, z \in \nat$. Let $a$ be an $n$-bit coefficient drawn according
to some discrete probability distribution bounded from above by $\phi^{1/n}$.
Then $\Pr(a \in [z, z+\delta)) \leq \phi^{1/n} \delta$.
\end{lemma}

\begin{proof}
There are exactly $\delta$ outcomes of $a$ that lead to $a \in [z, z+\delta)$.
Thus, $\Pr(a \in [z, z+\delta)) \leq \phi^{1/n} \eps$.
\end{proof}

Our goal is to show that $O(\log(n\phi^{1/n} 2^{n}))$ bits for each coefficient
suffice to determine whether a solution exists. (For the average case, we have
$\phi = 2^{-n^2}$, thus $O(\log n)$ bits per coefficient.)
To do this, it is not sufficient for an $x \in S$ to just satisfy $w^T x \leq t$:
Because of the rounding, we might find that $x$ is feasible with respect
to the rounded coefficients whereas $x$ is infeasible with respect
to the true coefficients.
Thus, what we need is that $w^Tx$ is sufficiently smaller than $t$. Then the rounding
does not affect the feasibility.
Unfortunately, we cannot rule out the existence of solutions $x \in S$ that are very close
to the threshold (after all, there can be an exponential number of solutions,
and it is likely that some of them are close to the threshold).
But it is possible to prove the following:
Assume that there is some ranking among the solutions $x \in S$.
Let the winner be the solution $x^\star \in S$ that fulfills $w^T x^\star \leq t$
and is ranked highest among all such solutions. Then it is likely
that $t - w^T x^\star$ is not too small.
Now, any solution that is ranked higher than $x^\star$
must be infeasible because it violates the linear constraint $w^T x \leq t$.
Let $\hat x$ be the solution that minimizes $w^T x - t$ among all solutions
ranked higher than $x^\star$. Then it is also
unlikely that $w^T \hat x - t$ extremely small, i.e.,
that $\hat x$ violates the linear constraint by only a small margin.

\begin{remark}
In Beier and V\"ocking's analysis~\cite{SmoothedBinary}, the ranking was given by the
objective function.
We do not have an objective function here because we deal with
decision problems.
Thus, we have to introduce a ranking artificially. In the following,
we use the lexicographic ordering (if not mentioned otherwise), which satisfies
the following \emph{monotonicity} property that simplifies the proofs:
if $x \in S$ is ranked higher than $y \in S$,
then there is an $i$ with $x_i = 1$ and $y_i = 0$.
\end{remark}

Now let $x^\star$ be the winner (if it exists), i.e., the highest ranked
(with respect to lexicographic ordering) solution among all feasible solutions. Then
we define the \emph{winner gap} as
\[
  \Gamma(t) = \begin{cases}
     t - w^T x^\star & \text{if there exists a feasible solution and} \\
     \bot & \text{otherwise.}
     \end{cases}
\]
The goal is to show that it is unlikely that $\Gamma$ is small.
In order to analyze $\Gamma$, it is useful to define also the loser gap $\Lambda$.
The loser $\hat x \in S$ is a solution that is ranked higher
than $x^\star$ but cut off by the constraint $w^T x \leq t$.
It is the solution with minimal $w^T x - t$ among all such solutions.
(If there is a tie, which can happen because we have discrete probability distributions, then
we take the highest-ranked solution as the loser.)
We define
\[
\Lambda(t) = \begin{cases}
    w^T \hat x -t & \text{if there exists a loser $\hat x$ and} \\
    \bot & \text{otherwise.}
    \end{cases}
\]
The probability that $\Lambda$ or $\Gamma$ is smaller than some value $\delta$ is bounded
by $\delta \phi^{1/n}n$, which we will prove in the following.

The following lemma states that it suffices to analyze
$\Lambda$ in order to get bounds for both $\Lambda$ and $\Gamma$.
In fact, for the setting with just one linear constraint with non-negative coefficients,
we do not even need the winner gap.
But the winner gap is needed for more general cases, which we discuss
in
Section~\ref{sec:binarydisc}
but do not treat in detail for conciseness.

\begin{lemma}[\mbox{discrete version of \cite[Lemma 7]{SmoothedBinary}}]
\label{lem:samegap}
For all $t$ and $\delta$, we have
$\Pr(\Gamma(t) < \delta) = \Pr(\Lambda(t-\delta) \leq \delta)$.
\end{lemma}

\begin{proof}
A solution $x \in S$ is called \emph{Pareto-optimal} if there is no other solution
$x' \in S$ such that $w^T x' \leq w^T x$ and $x'$ is ranked higher than $x$.
Let us make two observations.
First, we observe that both winners and losers are Pareto-optimal.
Second, for every Pareto-optimal solution $x$, there exists a threshold $t$
such that $x$ is the loser for this particular threshold. To see this, simply set
$t = w^T x - 1$.

Let $P \subseteq S$ be the set of Pareto-optimal solutions. Then
\begin{align*}
  \Gamma(t) & = \min\{t - w^T x \mid x \in P, w^T x \leq t\} \quad \text{and} \\
  \Lambda(t) & = \min\{w^T x - t \mid x \in P, w^T x > t\} = \min\{w^T x - t \mid x \in P, w^T x \geq t+1\}.
\end{align*}
Now $\Gamma(t) < \delta$ if and only if there is an $x \in P$ with $t-w^T x \in \{0, \ldots, \delta-1\}$.
This is equivalent to $w^T x - t \in \{-\delta + 1, \ldots, 0\}$
and to $w^T x - (t - \delta) \in \{1, \ldots, \delta\}$.
In turn, this is equivalent to $\Lambda(t - \delta) \leq \delta$.
\end{proof}

Now we analyze $\Lambda(t)$.
The following lemma makes this rigorous. It is a discrete counterpart to Beier and V\"ocking's separating
lemma~\cite[Lemma 5]{SmoothedBinary}.
We have to assume that the all-zero vector is not contained in $S$.
The reason for this is that its feasibility does not depend on any randomness.

\begin{lemma}[separating lemma]
\label{lem:separating}
Suppose that $(0,\ldots, 0) \notin S$. For every $\delta, t \in \nat$,
we have $\Pr(\Gamma(t) < \delta) \leq \delta \phi^{1/n} n$
and $\Pr(\Lambda(t) \leq \delta) \leq \delta \phi^{1/n} n$.

If we use a non-monotone ranking, then the bounds for the probabilities
become $\delta \phi^{1/n} n^2$.
\end{lemma}

\begin{proof}
Because of Lemma~\ref{lem:samegap}, it suffices to analyze the loser gap $\Lambda$.
We only give a proof sketch for monotone rankings as that emphasis the differences to the
continuous counterpart~\cite[Lemma 5]{SmoothedBinary}.

Let $S_i = \{x \in S \mid x_i = 1\}$, and let $\overline S_i = S \setminus S_i =
\{x \in S \mid x_i = 0\}$.
Let $x^{\star i} \in \overline S_i$ be the winner from $\overline S_i$:
$x^{\star i}$ is ranked highest in $\overline S_i$ and satisfies the linear
constraint $w^T x^{\star i} \leq t$.
Let $\hat x^i \in S_i$ be the loser with respect to $x^{\star i}$, i.e.,
a solution that is ranked higher than $x^{\star i}$ and minimizes
$w^T \hat x^i  - t$ (if such a solution exists).
Let
\[
  \Lambda_i = \begin{cases}
     w^T \hat x^i - t & \text{if $\hat x^i$ exists and} \\
     \bot & \text{otherwise}.
     \end{cases}
\]
Note that $\hat x^i$ can be feasible and, thus, $\Lambda_i$ can be negative.

To analyze $\Lambda_i$, we assume that all $w_j$ with $j \neq i$
are fixed by an adversary. The winner $x^{\star i}$ does not depend $w_i$
because all solutions $x \in \overline S_i$ have $x_i = 0$.
Once $x^{\star i}$ is fixed, also $\hat x^i$ is fixed.
Because $w_j$ for $j \neq i$ is fixed and $\hat x^i_i = 1$, we can rewrite
$w^T \hat x^i - t = z + w_i$.
Now $\Lambda_i \in \{1, \ldots, \delta\}$ if $w_i$ assumes a value
in some interval of length $\delta$, which happens
with a probability of at most $\delta \phi^{1/n}$.

Furthermore, if $\Lambda \neq \bot$, then there exists an $i$ with
$\Lambda_i = \Lambda$ \cite[Claim B]{SmoothedBinary}. Thus,
a union bound over all $n$ possibilities for $i$
yields $\Pr(\Lambda(t) \leq \delta) \leq \delta \phi^{1/n} n$.
\end{proof}

For the probabilistic constraint $w^T x \leq t$, it is not sufficient for
an $x$ to satisfy it. Instead, we want that only a few bits of each coefficient of $w$ suffice to
find an $x$ that satisfies that constraint.
Here, ``few'' means roughly $O(\log(n\phi^{1/n} 2^n))$. (Note that this is roughly $O(\log n)$ if we are
close to the average case, where $\phi \approx 2^{-n^2}$.)
Different from Beier and V\"ocking's continuous case (where the real-valued coefficients where
revealed by an oracle), we have the true coefficients at hand.
Thus, we do not need their certificates that a solution is indeed feasible,
but we can simply test with the true coefficients. Clearly, this testing can be done in polynomial time.

For an $n$-bit natural number $a$ and $b \in \nat$, let $\lfloor a \rfloor_b = $ be the number
obtained from $a$ by only taking the $b$ most significant bits. This means
that $\lfloor a \rfloor_b = 2^{n-b} \cdot \lfloor a/2^{n-b} \rfloor$.

In order to show that pseudo-polynomiality implies smoothed polynomial complexity,
we use a pseudo-polynomial algorithm as a black box in the following way:
We run the pseudo-polynomial algorithm with the highest $O(\log n)$ bits.
(To do this, we scale the rounded coefficients of $w$ down. Furthermore, we also
have to scale $t$ down appropriately.)
If we find a solution, then
we check it against the true coefficients of $w$. If it remains feasible, we output ``yes''. If it becomes
unfeasible, then we take one more bit for each coefficient and continue.
The following lemma gives a tail bound for how long this can go on.

\begin{lemma}
\label{lem:tailbits}
Assume that we use $b$ bits for each coefficient of $w$.
Let $x^\star$ be the winner (with respect to the true $w$ without rounding).
The probability that solving the problem with
$b$ bits for each coefficient yields a solution different from $x^\star$ is bounded from above
by $2^{n-b} \phi^{1/n} n^2$.
\end{lemma}

\begin{proof}
We only get a solution different from $x^\star$ if there is a solution $\hat x$ ranked higher than $x^\star$
that is feasible with respect to the rounded coefficients.
By rounding, we change each coefficient by at most $2^{n-b}$.
Thus, $w^T \hat x - \lfloor w \rfloor_b^T \hat x \leq 2^{n-b} n$.

We can conclude that we find $\hat x$ instead of $x^\star$ only if the loser gap $\Lambda$ is at most
$2^{n-b} n$, which happens with a probability of at most 
$2^{n-b} \phi^{1/n} n^2$
(or $2^{n-b} \phi^{1/n} n^3$ if the ranking is not monotone).
\end{proof}

With this preparation, we can prove the main result of this section.

\begin{theorem}
\label{thm:smoothedbinary}
If a binary decision problem can be solved in pseudo-polynomial time,
then it is in $\SmoothedP$.
\end{theorem}

\begin{proof}
We have to show
that the running time of the algorithm sketched above, which uses
the pseudo-polynomial algorithm as a black box,
fulfills Theorem~\ref{thm:alt1}.

If $b$ bits for each coefficient are used, the running time of the pseudo-poly\-nomial algorithm is
bounded from above by $O((n2^b)^c)$ for some constant $c$.
(Even the total running time summed over all iterations up to $b$ bits being revealed
is bounded by $O((n2^b)^c)$, because it is dominated by the last iteration.)

The probability that more than time $t = O((n2^b)^c)$ is needed is bounded from above
by $2^{n-b} \phi^{1/n} n^2$ according to Lemma~\ref{lem:tailbits}.
We can rewrite this as
\[
  2^{n-b} \phi^{1/n} n^2  = n^2 2^{-b} \bigl(2^{n^2} \phi\bigr)^{1/n}
   = \frac{n^3}{O(t^{1/c})} \cdot \bigl(2^{n^2} \phi\bigr)^{1/n}
   \leq \frac{n^3}{O(t^{1/c})} \cdot 2^{n^2} \phi .
\]
The last inequality holds because $\phi \geq 2^{-n^2}$.
The theorem is proved because this tail bound for the running time is strong
enough according to Theorem~\ref{thm:alt1}.
\end{proof}

\subsection{Examples and Discussion}
\label{sec:binarydisc}

Examples of problems in
\SmoothedP\ are the decision problems associated with the
following \NP-hard optimization problems:
\begin{itemize}
\item knapsack, where the goal is to find a subset of a given collection of items
   that maximizes the profit while obeying a budget for its weight;
\item constrained shortest path, where the goal is to find a path of minimum length
   that obeys a certain a budget;
\item constrained minimum-weight spanning tree.
\end{itemize}

These problems can be solved in pseudo-polynomial time using dynamic programming, even if we insist
on a lexicographically maximal solution (as we have to for Lemma~\ref{lem:separating}).

Let us now discuss some extensions of the model.
We have restricted ourselves to deterministic pseudo-polynomial algorithms, which
yield smoothed polynomial complexity. These deterministic
algorithms can be replaced without any complication
by randomized errorless algorithms that have expected pseudo-polynomial
running time.

So far, we have not explicitly dealt with constraints of the form ``$w^T x \geq t$''.
But they can be treated in the same way as ``$w^T x \leq t$'', except
that winner and loser gap change their roles.
Furthermore, we did not include the case that coefficients can be positive or negative.
This yields additional technical difficulties (we have to round more carefully and
take both winner and loser gap into account), but we decided to restrict ourselves
to the simpler form with non-negative coefficients for the sake of clarity.
Moreover, we have not considered the case of multiple linear constraints~\cite[Section 2.3]{SmoothedBinary}
for the same reason.
Finally, R\"oglin and V\"ocking~\cite{SmoothedInteger} have extended the smoothed analysis framework
to integer programming. We believe
that the same can be done for our discrete
setting.

The main open problem concerning \SmoothedP\ and integer optimization is the following:
Beier and V\"ocking~\cite{SmoothedBinary} have proved that (randomized) pseudo-polynomiality and
smoothed polynomiality are equivalent.
The reason why we do not get a similar result is as follows:
Our ``joint density'' for all coefficients is bounded by $\phi$,
and the density of a single coefficient is bounded by $\phi^{1/n}$.
In contrast, in the continuous version, the joint density is bounded by $\phi^n$
while a single coefficient has a density bounded by $\phi$.

However, our goal is to devise a general theory for arbitrary decision problems.
This theory should include integer optimization, but it should not be restricted to integer optimization.
The problem is that generalizing the concept of one distribution bounded by $\phi$ for each coefficient
to arbitrary problems involves knowledge about the instances and the structure of the specific problems.
This knowledge, however, is not available if we want to speak about classes of decision problems
as in classical complexity theory.

\section{Tractability 2: Graphs and Formulas}
\label{sec:graph}

\subsection{Graph Coloring and Smoothed Extension of $G_{n,p}$}

The perturbation model
that we choose is the \emph{smoothed extension of $G_{n,p}$}~\cite{SmoothedSurvey}:
Given an adversarial graph $G=(V,E)$ and an $\eps \in (0, 1/2]$,
we obtain a new graph $G' = (V,E')$ on the same set of vertices
by ``flipping'' each (non-)edge of $G$ independently with a probability
of $\eps$. This means the following: If $e = \{u,v\} \in E$, then
$e$ is contained in $E'$ with a probability of $1-\eps$.
If $e = \{u,v\} \notin E$, then $\Pr(e \in E') = \eps$.

Transferred to our framework, this means the following:
We represent a graph $G$ on $n$ vertices as a binary string of length $\binom n2$,
and we have $N_{n, G} = 2^{\binom n2}$.
The flip probability $\eps$ depends on $\phi$: We choose
$\eps \leq 1/2$ such that $(1-\eps)^{\binom n2} = \phi$.
(For $\phi = 2^{-\binom n2} = 1/N_{n, G}$, we have
a fully random graph with edge probabilities of $1/2$.
For $\phi = 1$, we have $\eps = 0$, thus the worst case.)

We will not present an exhaustive list of graph problems in \SmoothedP, but
we will focus
on graph coloring as a very simple example. $\gcol k$ is the decision problem
whether the vertices of a graph can be colored with $k$ colors such that
no pair of adjacent vertices get the same color. $\gcol k$ is \NP-complete for
any $k \geq 3$~\cite[GT 4]{GareyJohnson}.

\begin{theorem}
\label{thm:coloring}
For any $k \in \nat$, $\gcol k \in \SmoothedP$.
\end{theorem}

\begin{proof}
To show that $\gcol k \in \SmoothedP$, we analyze the following simple algorithm:
First, we check whether the input graph contains a clique of size $k+1$.
This can be done easily in polynomial time.
If yes, we output no. If no, we perform exhaustive search.

The analysis is similar to Wilf's analysis~\cite{Wilf} of the coloring problem:
First, we check whether the input graph contains a clique of size $k+1$.
This can be done easily in polynomial time.
If yes, we output no. If no, we perform exhaustive search. The correctness of the algorithm
is obvious.

A graph is $k$-colorable only if it does not contain a clique of size $k+1$.
The probability that a specific set of $k+1$ vertices
form a $k+1$ clique is at least $\eps^{\binom{k+1}2}$.
Thus, the probability that a graph $G$ on $n$ vertices does not contain
a $k+1$ clique is at most $\left(1 - \eps^{\binom{k+1}2}\right)^{\frac{n}{k+1}}$.

We distinguish two cases:
First, $\eps \geq 0.1$. In this case, $\left(1 - \eps^{\binom{k+1}2}\right)^{\frac{n}{k+1}}$
can be bounded from above by $c^n$ for some positive constant $c < 1$
that depends on $k$. Brute-force testing whether a graph can be $k$-colored
can be done in time $\poly(n) \cdot k^n$.
The probability that we need brute force is at most $c^n$.
Thus, the expected running-time, raised to the power $\eps = \log_{k}(1/c)$,
is bounded from above by a polynomial.

Second, $\eps < 0.1$. Then we have $\phi = (1-\eps)^{\binom n2} \geq 0.9^{\binom n2}$.
The allowed running-time (raised to some constant power) is
$N_{n, G} \phi n = 2^{\binom n2} \phi n \geq 1.8^{\binom n2}$.
Thus, we can afford exhaustive search in every run.
\end{proof}

\begin{remark}
Bohman et al.~\cite{BohmanEA} and Krivelevich et al.~\cite{KrivelevichEA:SmoothedGraphs:2006}
consider a slightly different model for perturbing graphs:
Given an adversarial graph, we add random edges to the graph to obtain our actual instance.
No edges are removed.

They analyze the probability that the random graph thus obtained is guaranteed
to contain a given subgraph $H$. By choosing $H$ to be a clique of size $k+1$ and using a proof
similar to Theorem~\ref{thm:coloring}'s,
we obtain that $\gcol k \in \SmoothedP$ also with respect to this perturbation model.
\end{remark}

\subsection{Unsatisfiability and $\SmoothedRP$}

Besides the smoothed extension of $G_{n,p}$ discussed above,
there exist various other models for obtaining graphs and also Boolean formulas
that are neither fully random nor adversarial.

Feige~\cite{Feige:RefutingSmoothedCNF:2007} and Coja-Oghlan et
al.~\cite{CojaOghlanEA:Walksat:2009} have considered the following model:
We are given a (relatively dense) adversarial Boolean $k$-CNF formula. Then
we obtain our instance by negating each literal with a small probability.
It is proved that such smoothed formulas are likely to be unsatisfiable, and that
their unsatisfiability can be proved efficiently. However, their algorithms
are randomized, thus we do not get a result that \kunsat\ (this means that
unsatisfiability problem for $k$-CNF formulas) for dense instances belongs to
\SmoothedP. However, it shows that $\kunsat$ for dense instance belongs
to \SmoothedRP, where $\SmoothedRP$ is the smoothed analogue of $\RP$:
A pair $(L, \Dist)$ is in $\SmoothedRP$ if there is a randomized polynomial
algorithm $A$ with the following properties:
\begin{enumerate}
\item For all $x \notin L$, $A$ outputs ``no''.
(This property is independent of the
   perturbation.)
\item For all $x \in L$, $A$ outputs ``yes'' with a probability of at least $1/2$.
 (This property is also independent of the
   perturbation.)
\item $A$ has smoothed polynomial running time with respect to $\Dist$.
(This property is independent of the internal randomness of $A$.)
\end{enumerate}
Note that we have two sources of randomness in $\SmoothedRP$: The instance
is perturbed, and the algorithm $A$ is allowed to use randomness. Item~1 and~2 depend only
on $A$'s own randomness. Item~3 depends only on the perturbation $\Dist$.

Now, let $\kunsat_{\beta}$ be $\kunsat$ restricted to instances
with at least $\beta n$ clauses, where $n$ denotes the number of variables.
Let $\eps$ be the probability that a particular literal is negated.
Feige~\cite{Feige:RefutingSmoothedCNF:2007} has presented a polynomial-time algorithm
with the following property: If $\beta = \Omega(\sqrt{n \log\log n}/\eps^2)$
and the perturbed instance of $\kunsat_{\beta}$ is unsatisfiable,
which it is with high probability, then his algorithm proves that the formula is unsatisfiable with
a probability of at least $1-2^{\Omega(-n)}$.
The following result is a straightforward consequence.

\begin{theorem}
$\kunsat_{\beta} \in \SmoothedRP$ for $\beta = \Omega(\sqrt{n \log \log n})$.
\end{theorem}

\section{Smoothed Analysis vs.\ Semi-Random Models}
\label{sec:semirandom}

Semi-random models for graphs and formulas exist even longer than smoothed analysis and can be considered
as precursors to smoothed analysis. The basic concept is as follows: Some instance is created randomly that
possesses a particular property. This property can, for instance, be that the graph is $k$-colorable.
After that, the adversary is allowed to modify the instance without destroying the property.
For instance, the adversary can be allowed to add arbitrary edges between the different color classes.
Problems that have been considered in this model or variants thereof
are independent set \cite{FeigeKilian:Semirandom:2001},
graph coloring~\cite{BlumSpencer:Coloring:1995,FeigeKilian:Semirandom:2001,CojaOghlan:ColoringSemirandom:2007},
or finding sparse induced subgraphs~\cite{CojaOghlan:Semirandom:2007}.
However, we remark that these results do not easily fit into a theory of smoothed analysis. The reason is that
in these semi-random models, we first have the random instance, which is then altered by the adversary.
This is in contrast to smoothed analysis in general and our smoothed complexity theory in particular,
where we the adversarial decisions come before the randomness is applied.

\section{Discussion}
\label{sec:disc}

Our framework has many of the characteristics
that one would expect.
We have reductions and complete problems
and they work in the way one expects them to work.
To define reductions, we have to use the concept
of disjoint supports. It seems to be essential that we know the original instance $x$
that the actual instance $y$ was drawn from
to obtain proper domination.
Although this is somewhat unconventional,
we believe that this is the right way to define
reductions in the smoothed setting. The reason is that
otherwise, we do not know the probabilities of the instances, which
we need in order to apply the compression function. The compression function, in turn,
seems to be crucial to prove hardness results.
Still, an open question is whether a notion of reducibility can be defined
that circumvents these problems.
Moreover, many of the positive results from smoothed analysis
can be cast in our framework, like it is done
in Sections~\ref{sec:trac1} and \ref{sec:graph}.

Many positive results in the literature
state their bounds in the number of ``entities''
(like number of nodes, number of coefficients) 
of the instance. However, in complexity theory,
we measure bounds in the length (number of symbols)
of the input in order to get a theory for arbitrary problems,
not only for problems of a specific type. To state bounds in terms of bit length makes things less tight,
for instance the reverse direction of integer programming does not work.
But still, we think it is more important and useful to use the usual notion
of input length such that smoothed complexity fits with average-case and worst-case complexity.

Finally, the results by R\"oglin and Teng~\cite{RoeglinTeng} show that, for binary optimization problems,
expected polynomial is indeed a robust measure. We have shown that this is in general not the case. To do this,
we have used a language in $\EXPLIN$. The obvious question is now whether $\AvgP$ and $\AvgPBM$ as well
as $\SmoothedP$ and $\SmoothedPBM$ coincide for problems in $\NP$.

We hope that the present work will stimulate further
research in smoothed complexity theory in order to
get a deeper understanding of the theory behind
smoothed analysis.



\end{document}